\documentclass[journal]{IEEEtran}
\IEEEoverridecommandlockouts
\usepackage{graphicx}
\usepackage{enumitem}
\usepackage{subcaption}
\usepackage{cite}
\usepackage{amsmath,amssymb,amsfonts}
\usepackage{algorithmic}
\usepackage{graphicx}
\usepackage{textcomp}
\usepackage{xcolor}
\usepackage{amsmath}
\usepackage{amssymb}
\usepackage{bm}
\usepackage{textcomp}
\usepackage{graphicx}
\usepackage{xcolor}
 \usepackage{amsmath}
\usepackage{amsthm}
\usepackage{url}
\usepackage{mathtools}
\usepackage{tikz}
\newcommand\myline[1][]{%
   \,\tikz[baseline]\draw[thick,#1](0,-\dp\strutbox)--(0,\ht\strutbox);\,%
}
\usepackage{multirow}
\DeclarePairedDelimiter\ceil{\lceil}{\rceil}
\DeclarePairedDelimiter\floor{\lfloor}{\rfloor}
\theoremstyle{definition}

\def\BibTeX{{\rm B\kern-.05em{\sc i\kern-.025em b}\kern-.08em
    T\kern-.1667em\lower.7ex\hbox{E}\kern-.125emX}}
\newtheorem*{remark}{Remark}
\newtheorem{prop}{Proposition}

\begin{document}

\title{Index and Composition Modulation\\
\thanks{The authors wish to acknowledge the support of the Bristol Innovation $\&$ Research Laboratory of Toshiba Research Europe Ltd. }
}

\author{Ferhat Yarkin~\IEEEmembership{Graduate~Student~Member,~IEEE} and Justin P.~Coon~\IEEEmembership{Senior~Member,~IEEE}
\thanks{F. Yarkin and J. P. Coon are with the Department of Engineering Science, University of Oxford, Parks Road, Oxford, OX1 3PJ, U.K. E-mail: \{ferhat.yarkin and justin.coon\}@eng.ox.ac.uk}}

\maketitle

\begin{abstract}
In this paper, we propose a novel modulation concept which we call \emph{index and composition modulation (ICM)}. In the proposed concept, we use indices of active/deactive codeword elements and compositions of an integer to encode information. In this regard, we first determine the activated codeword elements, then we exploit energy levels of these elements to identify the compositions. We depict a practical scheme for using ICM with orthogonal frequency division multiplexing (OFDM) and show that OFDM with ICM (OFDM-ICM) can enhance the spectral efficiency (SE) and error performance of OFDM-IM. We design an efficient low-complexity detector for the proposed technique. Moreover, we analyze the error and SE performance of the OFDM-ICM technique and show that it is capable of outperforming existing OFDM benchmarks in terms of error and SE performance.
\end{abstract}


\section{Introduction}

Index modulation (IM) that encodes data into the combinations of active/deactive codeword elements offers a variety of attractive advantages including better error performance and improved energy/spectral efficiency over conventional modulation and multiplexing schemes. Hence, the adaptation of IM to  orthogonal  frequency  division  multiplexing  (OFDM) has attracted several researchers’ attentions \cite{Basar2013, Fan2015, Mao2017,Wen2017, Yarkin2020set, Yarkin2020} and it has been shown by these studies that OFDM with IM (OFDM-IM) can achieve better error performance, higher data rate and higher energy efficiency than conventional OFDM. Apart from IM, in \cite{Yarkin2020comp}, we proposed two new modulation concepts, \emph{weak composition modulation (WCM)} and \emph{composition modulation (CM)}, that embed data using integer compositions, and showed that the applications of these concepts to OFDM bring noteworthy improvements in error performance. 

% OFDM-IM encodes information in the combinations of activated/deactivated subcarriers \cite{Basar2013, Fan2015}, whereas the indices of distinguishable constellations are used in \cite{Mao2017,Wen2017, Yarkin2020set, Yarkin2020} to encode information and  achieve higher spectral efficiency (SE) and better error performance than conventional OFDM-IM. Other IM-based applications exist with manifestations in space, time and frequency  (see, e.g., \cite{Mao2019}). Moreover, in \cite{Yarkin2020comp}, we proposed two new modulation concepts, \emph{weak composition modulation (WCM)} and \emph{composition modulation (CM)}, that embed data using integer compositions, and showed that the applications of these concepts to OFDM bring noteworthy improvements in error performance compared to OFDM-IM and OFDM. 

5G's key enabling technologies are expected to satisfy the user demands for high data rate, ultra-reliable transmission, and very low latency. However, it is anticipated that 6G and the following next-generation networks will require even higher data rates, better error performance, and lower latency due to the proliferation of a variety of new applications, including extended reality services, telemedicine, haptics, flying vehicles, brain-computer interfaces, and connected autonomous systems \cite{Saad2019}. In this regard, the existing OFDM and OFDM-IM schemes require further improvement in error performance and spectral efficiency (SE) to be deployed efficiently in next-generation networks.

Motivated by the advantages of the IM and CM techniques as well as the requirements of next-generation networks, we propose a novel concept that we call \emph{index and composition modulation (ICM)}. In this context, we encode information using the combinations of active/deactive codewords elements as in IM and the combinatorial framework of integer compositions as in CM. We depict a practical model based on OFDM with ICM (OFDM-ICM). Moreover, we design an efficient low-complexity detector to overcome the high complexity arising from maximum-likelihood (ML) detector. We also investigate the bit error rate (BER) and SE of the proposed scheme in this paper. Our analytical, as well as numerical, findings indicate that our novel design can achieve a substantially better SE and BER  performance than OFDM, OFDM-IM, OFDM-WCM and OFDM-CM. 

\section{Index and Composition Modulation}\label{sec:secII}

An ICM codebook consists of $L_{ICM}$ codewords and each codeword is an $N$-tuple of nonnegative real numbers which can be regarded as a vector $\textbf{x}_l=\big\{x_{l1}, x_{l2}, \ldots, x_{lN}\big\}$, $l=1, 2, \ldots, L_{ICM}$, in an Euclidean space $\mathcal{S}$ of $N$ dimensions where $x_{ln} \in \mathbb{R}_{\ge 0}$ and $n=1, 2, \ldots, N$. In an ICM codeword, $K$ elements are positive real numbers where $K \le N$, whereas the remaining $N-K$ elements are zeros. Here, we call the $K$ positive real elements and $N-K$ zero elements as activated and deactivated elements, respectively, since no energy is used to form the zero elements. All other codewords of the codebook can be obtained by permuting the order of the $N$ elements in the codeword. The values of $K$ activated elements in each codeword are chosen according to the integers that form the compositions of an integer $I$ with $K$ parts. Hence, one can obtain different codewords by using different compositions of an integer. More explicitly, when we form $\textbf{x}_l$, we first determine the indices of $K$ activated elements in an $N$-tuple and denote the set of these indices as $\mathcal{I} \coloneqq \big\{\alpha_1, \alpha_2, \ldots, \alpha_K\big\}$ where $\alpha_k \in \big\{1, 2, \ldots, N\big\}$ and $k=1, 2, \ldots, K$. Then, the $k$th activated element $x_{l\alpha_k}$ of the $N$-tuple is chosen as $x_{l\alpha_k}=\sqrt{\nu_k}$ where $\nu_k\in \big\{ 1, \ldots, I-K+1\big\}$ is the $k$th summand of the composition and $\nu_1+\nu_2+\ldots+\nu_K=I$. Note that one needs to pick $I \ge K$ to be able to construct the ICM codewords. Since the number of permutations regarding the order of $K$ activated and $N-K$ deactivated elements is  $\binom{N}{K}$ and the number of integer compositions of $I$ with $K$ parts is $\binom{I-1}{K-1}$, the number of codewords in an ICM codebook is $L_{ICM}=\binom{N}{K}\binom{I-1}{K-1}$. In Table \ref{tab:table0}, we give a codebook generation example for an ICM scheme when $N=I=3$ and $K=2$. As seen from the leftmost column of the table, we first determine the indices of $K=2$ activated elements, then we determine the compositions of an integer $I=3$ with $K=2$ parts and map the summands in those compositions to the elements of the ICM codewords according to the indices as shown in the rightmost column.  

% Just like IM, variations or applications of this encoding scheme can be used in power-line, time-slotted, frequency-slotted and space-slotted communications. In this context, we show a frequency domain application of ICM in the following section.  

% \usepackage{multirow}
\begin{table}[t!]
\centering
\caption{Codebook generation example for ICM when $N=I=3$ and $K=2$.}
\label{tab:table0}
\begin{tabular}{|c|c|c|}
\hline
\begin{tabular}[c]{@{}c@{}}Indices of Activated \\ Elements, $\mathcal{I}$\end{tabular} & Compositions & ICM Codeword \\ \hline
$\big\{1, 2\big\}$ & \multirow{3}{*}{3=1+2} & $\textbf{x}_1=\big\{1, \sqrt{2}, 0\big\}$ \\ \cline{1-1} \cline{3-3} 
$\big\{1, 3\big\}$ &  & $\textbf{x}_2=\big\{1,  0, \sqrt{2}\big\}$ \\ \cline{1-1} \cline{3-3} 
$\big\{2, 3\big\}$ &  & $\textbf{x}_3=\big\{0, 1, \sqrt{2}\big\}$ \\ \hline
$\big\{1, 2\big\}$ & \multirow{3}{*}{3=2+1} & $\textbf{x}_4=\big\{\sqrt{2}, 1, 0\big\}$ \\ \cline{1-1} \cline{3-3} 
$\big\{1, 3\big\}$ &  & $\textbf{x}_5=\big\{\sqrt{2}, 0, 1\big\}$ \\ \cline{1-1} \cline{3-3} 
$\big\{2, 3\big\}$ &  & $\textbf{x}_6=\big\{0, \sqrt{2}, 1\big\}$ \\ \hline
\end{tabular}
\end{table}

\section{OFDM with Index and Composition Modulation}\label{sec:secIII}

In this section, we present a practical system model in which we apply the ICM concept to OFDM transmissions.

\subsection{Transmitter}

Analogously to OFDM-IM, $m$ input bits enter the transmitter and these bits are divided into $B=m/f$ blocks, each of them having $f$ input bits. The total number of the subcarriers, $N_T$, is also split into  $B=N_T/N$ blocks whose size is $N$. We will focus on the $b$th block in what follows. $f$ information bits are further divided into three parts, having $f_1$, $f_2$ and $f_3$ bits with $f_1+f_2+f_3=f$ in the $b$th block.

The first  $f_1=\floor{\log_2\binom{N}{K}}$ bits are used to determine the activated $K$ subcarriers. Then, the $f_2=\floor{\log_2\binom{I-1}{K-1}}$ bits are used to determine the specific composition of an integer $I$ with $K$ parts where $I \ge K$. The energies of the symbols on the activated subcarriers are chosen according to the specific composition of an integer $I$ with $K$ parts. Let us denote the sets that comprise  the indices of the activated subcarriers and the energies of the activated subcarriers in the $b$th block, respectively, as $\mathcal{I}^b\coloneqq\big\{\alpha_1, \alpha_2, \ldots, \alpha_K \big\}$ and $\beta^b\coloneqq\big\{\nu_1E_T/I,\nu_2E_T/I, \ldots, \nu_KE_T/I\big\}$ where $\alpha_k \in \big\{1, 2, \ldots, N\big\}$, $k=1, 2, \ldots, K$, $\nu_k\in\big\{1, 2, \ldots, I-K+1\big\}$, and $\nu_1+\nu_2+\ldots+\nu_K=I$  . Once we decide $\mathcal{I}^b$ and $\beta^b$ according to the $f_1$ and $f_2$ bits, respectively, $f_3=K\log_2 M$ bits are used to determine the $M$-PSK constellation symbols carried by the activated subcarriers. Hence, the SE of the OFDM-ICM scheme per subcarrier can be given as 
\begin{align}\label{eq:eq2last}
    \eta=\frac{f_1+f_2+f_3}{N}=\frac{\floor{\log_2\binom{N}{K}}+\floor{\log_2\binom{I-1}{K-1}}+K\log_2M}{N}.
\end{align}

The mapping of $f_1$ bits to the indices of activated subcarriers and $f_2$ bits to the subcarrier's energies can be implemented by using a look-up table. In Table \ref{tab:table2}, we present an example of how these mappings are performed when $N=I=4$ and $K=3$. Since $f_1=\floor{\log_2\binom{4}{3}}=2$, 
the first two bits, $p_1$, entering the OFDM-ICM encoder of the corresponding bitstream $[p_2~p_1]$ are used to determine the indices of activated subcarriers. Then, the remaining $f_2=\floor{\log_2 \binom{I-1}{K-1}}=\floor{\log_2 \binom{3}{2}}=1$ bit, $p_2$, is used to determine the specific composition of $I=4$ with $K=3$ parts. For example, when $[p_2~ p_1]=[0~ 1~ 1]$ bits enter the OFDM-ICM transmitter, the first two bits '$p_1=11$' choose the indices of activated subcarriers as $\mathcal{I}^b\coloneqq\big\{2, 3, 4\big\}$, then the remaining bit '$p_2=0$' chooses the set $\beta^b\coloneqq\big\{E_T/4, E_T/4, 2E_T/4\big\}$ for the subcarriers' energies. Then, we map the set, $\beta^b$, to the energies of the activated subcarriers and obtain the OFDM-ICM symbol vector as shown in the table.

% \usepackage{multirow}
% \usepackage{graphicx}
\begin{table}[t!]
\centering
\caption{Look-up table implementation example for OFDM-ICM when $N=I=4$ and $K=3$. }
\label{tab:table2}
\resizebox{\columnwidth}{!}{%
\begin{tabular}{|c|l|l|c|c|}
\hline
Compositions & \multicolumn{1}{c|}{$\beta^b$} & \multicolumn{1}{c|}{$\mathcal{I}^b$} & \begin{tabular}[c]{@{}c@{}}OFDM-ICM\\ Symbol Vector\end{tabular} & $[p_2\myline[dashed] ~p_1]$ \\ \hline
\multirow{4}{*}{4=1+1+2} & \multirow{4}{*}{$\big\{\frac{E_T}{4}, \frac{E_T}{4}, \frac{2E_T}{4}\big\}$} & $\big\{1, 2, 3\big\}$ & $[\frac{E_T}{4} ~\frac{E_T}{4} ~\frac{2E_T}{4} ~0]$ & [0~\myline[dashed]0 0] \\ \cline{3-5} 
 &  & $\big\{1, 2, 4\big\}$ & $[\frac{E_T}{4} ~\frac{E_T}{4} ~0 ~\frac{2E_T}{4}]$ & [0~\myline[dashed]0 1] \\ \cline{3-5} 
 &  & $\big\{1, 3, 4 \big\}$ & $[\frac{E_T}{4} ~0 ~\frac{E_T}{4} ~\frac{2E_T}{4}]$ & [0~\myline[dashed]1 0] \\ \cline{3-5} 
 &  & $\big\{2, 3, 4\big\}$ & $[0 ~\frac{E_T}{4} ~\frac{E_T}{4} ~\frac{2E_T}{4}]$ & [0~\myline[dashed]1 1] \\ \hline
\multirow{4}{*}{4=1+2+1} & \multirow{4}{*}{$\big\{\frac{E_T}{4}, \frac{2E_T}{4}, \frac{E_T}{4}\big\}$} & $\big\{1, 2, 3\big\}$ & $[\frac{E_T}{4} ~\frac{2E_T}{4} ~\frac{E_T}{4} 0]$ & [1~\myline[dashed]0 0] \\ \cline{3-5} 
 &  & $\big\{1, 2, 4\big\}$ & $[\frac{E_T}{4} ~\frac{2E_T}{4} ~0 ~\frac{E_T}{4}]$ & [1~\myline[dashed]0 1] \\ \cline{3-5} 
 &  & $\big\{1, 3, 4 \big\}$ & $[\frac{E_T}{4} ~0 ~\frac{2E_T}{4} ~\frac{E_T}{4}]$ & [1~\myline[dashed]1 0] \\ \cline{3-5} 
 &  & $\big\{2, 3, 4\big\}$ & $[0 ~\frac{E_T}{4} ~\frac{2E_T}{4} ~\frac{E_T}{4}]$ & [1~\myline[dashed]1 1] \\ \hline
\multirow{4}{*}{4=2+1+1} & \multirow{4}{*}{$\big\{\frac{2E_T}{4}, \frac{E_T}{4}, \frac{E_T}{4} \big\}$} & $\big\{1, 2, 3\big\}$ & $[\frac{2E_T}{4} ~\frac{E_T}{4} ~\frac{E_T}{4} ~0]$ & unused \\ \cline{3-5} 
 &  & $\big\{1, 2, 4\big\}$ & $[\frac{2E_T}{4} ~\frac{E_T}{4} ~0 ~\frac{E_T}{4}]$ & unused \\ \cline{3-5} 
 &  & $\big\{1, 3, 4 \big\}$ & $[\frac{2E_T}{4} ~0 ~\frac{E_T}{4} ~\frac{E_T}{4}]$ & unused \\ \cline{3-5} 
 &  & $\big\{2, 3, 4\big\}$ & $[0 ~\frac{2E_T}{4} ~\frac{E_T}{4} ~\frac{E_T}{4}]$ & unused \\ \hline
\end{tabular}%
}
\end{table}

Once we decided the activated subcarriers and their energies according to the first $f_1+f_2$ bits,\footnote{$f_1+f_2$ bits can be mapped to the indices of activated subcarriers and the compositions without using a look-up table implementation since the mapping of $f_1$ bits to the indices of activated subcarriers and $f_2$ bits to the compositions, thus to the subcarriers' energies, can be performed without a look-up table as discussed in \cite{Basar2013} and \cite{Yarkin2020comp}, respectively.} then the remaining $f_3$ bits are used to modulate the signals on the activated subcarriers by using an $M$-PSK constellation, $\mathcal{M}$. Hence, in the $b$th block, the OFDM-ICM symbol vector can be written as $\textbf{x}^b= [x_1^b, x_2^b, \ldots, x_N^b]$ where $x_i^b\in \big\{\emptyset\big\} \cup \mathcal{M}$. The energy of the symbol carried by the $k$th activated subcarrier of the $b$th block is $E_k^b=|x_k^b|^2=\nu_kE_T/I$ where $k=1, 2, \ldots, K$. After obtaining symbol vectors of all blocks, the overall OFDM-ICM vector is formed as  $\textbf{x}\coloneqq [x(1), x(2), \ldots,x(N_T)]^T=[\textbf{x}^1,\ldots, \textbf{x}^b, \ldots, \textbf{x}^B]^T\in \mathcal{C}^{N_T\times 1}$. After this point, exactly the same operations as conventional OFDM are applied.\footnote{We assume that the elements of $\textbf{x}$ are interleaved sufficiently and the maximum spacing is achieved for the subcarriers.}

\begin{remark}
OFDM-ICM is equivalent to OFDM-IM and OFDM when $I=K$ and $I=K=N$, respectively. It is clear from  \eqref{eq:eq2last} that the proposed scheme is capable of providing a higher SE than OFDM-IM and OFDM-CM. It is also important to note that OFDM-ICM is fundamentally different from the OFDM-IM schemes in \cite{Mao2017,Wen2017, Yarkin2020set, Yarkin2020} due to the use of signal levels to encode information into compositions rather than different signals to encode information into index patterns or permutations.
\end{remark}

\subsection{Receiver}

At the receiver, the received signal is down-converted, and the cyclic prefix is then removed from each received baseband symbol vector before processing with an FFT. After employing an $N_T$-point FFT operation, the frequency-domain received signal vector can be written as 
\begin{align}\label{eq:eq2}
    \textbf{y} \coloneqq [y(1),y(2), \ldots, y(N_T)]^T=\textbf{X}\textbf{h}+\textbf{n}
\end{align}
where  $\textbf{X}=\text{diag}(\textbf{x})$. Moreover, $\textbf{h}$ and $\textbf{n}$ are $N_T\times 1$ channel and noise vectors, respectively. Elements of $\textbf{n}$ follow the complex-valued Gaussian distribution $\mathcal{CN}(0,N_0)$ where $N_0$ is the noise variance.

Since the encoding procedure for each block is independent of others, decoding can be performed independently at the receiver. Hence, using ML detection, the detected symbol vector for the $b$th block can be written as 
\begin{align}\label{eq:eq3}
    ({\hat{\mathcal{I}}^b, \hat{\beta}^b, \hat{\textbf{x}}^b})= \arg \min_{\mathcal{I}^b, \beta^b, \textbf{x}^b} ||\textbf{y}^b-\textbf{X}^b\textbf{h}^b||^2
\end{align}
where $\textbf{y}^b=[y((b-1)N+1), \ldots, y(bN)]^T$, $\textbf{X}^b=\text{diag}(\textbf{x}^b)$ and $\textbf{h}^b=[h((b-1)N+1), \ldots, h(bN)]^T$.

The optimum ML detector in \eqref{eq:eq3} performs $2^{\floor{\log_2\binom{N}{K}}+\floor{\log_2\binom{I-1}{K-1}}}M^K$ squared Euclidean distance calculations. To overcome the high complexity of the optimum ML detector,  we design the following low-complexity ML (LC-ML) detector for the OFDM-ICM technique:

\begin{enumerate}[leftmargin=*]
    \item Calculate the following log-likelihood ratio (LLR) for the $n$th subcarrier\footnote{Here, we drop the block superscript, $b$, for convenience since the proposed LC-ML detector can be operated for each block independently.} where $n\in\big\{1, 2, \ldots, N\big\}$:
    \begin{align}\label{eq:eq5}
        \delta(n)&=\ln \frac{\sum\limits_{x\in\mathcal{M}}\sum\limits_{E\in\Lambda}P(x(n)=\sqrt{E}x|y(n))}{P(x(n)=\emptyset|y(n))} \\ & \nonumber \propto \ln \frac{\sum\limits_{x\in\mathcal{M}}\sum\limits_{E\in\Lambda}\exp{(-|y(n)-h(n)\sqrt{E}x|^2/N_0)}}{\exp(-|y(n)|^2/N_0)}
    \end{align}
    where $x$ is a unit-energy $M$-PSK symbol.  $P(x(n)=\sqrt{E}x|y(n))$ stands for the probability that the symbol carried by the $n$th subcarrier, $x(n)$, is equal to $\sqrt{E}x$ given the received signal regarding the $n$th subcarrier, $y(n)$, whereas $P(x(n)=\emptyset|y(n))$ is the probability that the $n$th subcarrier is deactivated given $y(n)$.  Moreover, $\Lambda$ is the set which consists of the possible energy levels, i.e., $\Lambda\coloneqq\big\{E_T/I, 2E_T/I,\ldots, (I-K+1)E_T/I\big\}$. Note that the LLR value in \eqref{eq:eq5} can be approximated by \cite{Wen2017}
    \begin{align}
        \delta(n) \approx \frac{\exp{(-|y(n)-h(n)\sqrt{\hat{E}(n)}\hat{x}(n)|^2/N_0)}}{\exp(-|y(n)|^2/N_0)} 
    \end{align}
    where $(\hat{x}(n), \hat{E}(n))=\arg \min_{x,E} |y(n)-h(n)\sqrt{E}x|^2$.
    \item Sort the LLR values  $(\delta(1), \delta(2), \ldots,\delta(N))$ in descending order, i.e., $\delta(\rho_1)>\delta(\rho_2)>\ldots>\delta(\rho_N)$ where $\rho_n\in \big\{1, 2, \ldots, N\big\}$. Then, determine the indices of the activated subcarriers as $\hat{\mathcal{I}}=\big\{\rho_1, \rho_2, \ldots, \rho_K\big\}$.
    \item Sort the channel gains of the activated subcarriers in descending order, i.e., $|h(\gamma_1)|^2>|h(\gamma_2)|^2>\ldots>|h(\gamma_K)|^2$ where $\gamma_k \in \hat{\mathcal{I}}$ and $\Gamma\coloneqq\big\{\gamma_1, \gamma_2, \ldots, \gamma_K\big\}$.
    \item Determine the energy levels and $M$-ary symbols on the activated $K$ subcarriers by following the order in $\Gamma$ and using ML detection. More explicitly, we start with the $\gamma_1$th subcarrier and determine its energy level and $M$-ary symbol by $(\hat{x}(\gamma_1), \hat{E}(\gamma_1))=\arg \min_{x \in \mathcal{M},E \in \Lambda} |y(\gamma_1)-h(\gamma_1)\sqrt{E}x|^2$. The related composition can be obtained by $\hat{\nu}(\gamma_1)=\hat{E}(\gamma_1)I/E_T$. Then, we update $\Lambda$ according to previously estimated energy levels and the fact that $\hat{\nu} (\gamma_k) \ge 1$ as $\Lambda\coloneqq\big\{E_T/I, \ldots, (I-(\hat{\nu}(\gamma_1)+K-2))E_T/I\big\}$ and we proceed with the $\gamma_2$th subcarrier and determine its energy level and $M$-ary symbol. For the $\gamma_k$th, $1<k<K$, and $\gamma_K$th subcarriers, the sets in question can be written respectively as $\Lambda\coloneqq\big\{E_T/I, \ldots, (I-(K-k+{\Delta}_k))E_T/I\big\}$ and $\Lambda\coloneqq\big\{ (I-{\Delta}_K)E_T/I\big\}$  where ${\Delta}_k=\sum_{l=1}^{k}\hat{\nu}(\gamma_l)$ and ${\Delta}_K=\sum_{l=1}^{K}\hat{\nu}(\gamma_l)$.
 \end{enumerate}

For the proposed detector, one needs to perform $(I-K+1)M$ squared Euclidean distance calculations in \eqref{eq:eq5}. Hence, the proposed algorithm makes $N(I-K+1)M$ squared Euclidean distance calculations in total, which is considerably smaller than the complexity of the optimum ML detector.

\section{Performance Analysis}\label{sec:secIV}
\subsection{Bit-Error Rate}

An upper-bound on the average BER is given by the well-known union bound as follows
\begin{align}
    P_b \leq \frac{1}{f2^f}\sum_{i=1}^{2^f}\sum_{j=1}^{2^f}P(\textbf{X}^i\to\textbf{X}^j)D(\textbf{X}^i\to\textbf{X}^j)
\end{align}
where $P(\textbf{X}^i\to\textbf{X}^j)$ is the pairwise error probability (PEP) regarding the erroneous detection of $\textbf{X}^i$ as $\textbf{X}^j$ where $i \neq j$, $i,j\in \big\{1, \ldots,L\big\}$, $\textbf{X}^i=\text{diag}(\textbf{x}^i)$, $\textbf{X}^j=\text{diag}(\textbf{x}^j)$, and $D(\textbf{X}^i\to\textbf{X}^j)$ is the number of bits in error for the corresponding pairwise error event. Here, $L$ is  the codebook size for the proposed scheme. One can use the same PEP expression as in \cite{Basar2013} and substitute the codewords of the proposed scheme to obtain the upper bound on the average BER.   

\begin{remark}
The minimum Hamming distance between the sets that keep the subcarriers' energies is two, just like the index symbols of OFDM-IM. However, the proposed scheme can send conventional modulation symbols together with embedding information into the indices and/or compositions, and the minimum Hamming distance between the conventional modulation symbols is limited to one. This limits the diversity gain of the proposed scheme to one. 
\end{remark}

As a special case of OFDM-ICM, one can use only the indices and compositions to embed information and obtain a diversity gain.\footnote{For this special case, the SE is $\eta=\frac{\floor{\log_2\binom{N}{K}}+\floor{\log_2\binom{I-1}{K-1}}}{N}$.} However, in this case, the transmitted symbol vectors should be carefully chosen to get optimum error performance. One simple yet efficient strategy is choosing the transmitted symbols in a way that the angular difference between the different energy symbols is maximized. Hence, in this case, one can choose the signal on the $k$th activated subcarrier of the $b$th block as  $x_k^b=\sqrt{\nu_kE_T/I}\exp(\frac{2j\nu_k\pi}{I-K+1})$. As will be shown in the Numerical Results section, the error performance of this special case is promising for low SEs.

% On the other hand, as shown in the Numerical Results Section of \cite{Yarkin2020}, one can send only index symbols and/or embed data into the energy domain to obtain BER curves whose diversity order is two.  

\subsection{Spectral Efficiency}\label{sec:secivb}
By rewriting $K$ and $I$ as $K=\alpha N$ and $I=\beta N$, respectively, the SE maximization problem is formulated as 
\begin{equation}\label{eq:eqopt}
    \begin{aligned}
    &\underset{\alpha}{\text{maximize}} && \eta=\frac{\log_2\binom{N}{\alpha N}+\log_2\binom{\beta N-1}{\alpha N-1}+\alpha N\log_2 M}{N} \\
    &\text{subject to} && 0 < \alpha \le 1, \\
    &&& \alpha \le \beta.
\end{aligned}
\end{equation}

\begin{prop}
The value of $\alpha$ that maximizes \eqref{eq:eq7} is  
\begin{align}\label{eq:eq8}
    \alpha^{*}=\frac{M(\beta+1) - \sqrt{M^2(\beta-1)^2+4M\beta}}{2(M-1)}.
\end{align}
Thus, one can pick the optimum number of activated subcarriers in an OFDM block as 
\begin{align}\label{eq:eq9}
    K^{*} \in\big\{ \floor{\alpha^{*}N}, \ceil{\alpha^{*}N}\big\}
\end{align}
where $\floor{.}$ and $\ceil{.}$ are floor and ceiling operations, respectively. 
\end{prop}

\begin{proof}
For $0 < \alpha_1< \alpha_2 \le 1$, one can write
\begin{align}\label{eq:eq7}
    \frac{\Delta \eta}{\Delta \alpha}&=\frac{\log_2\binom{N}{\alpha_2 N}+\log_2\binom{\beta N-1}{\alpha_2 N-1}+\alpha_2 N\log_2 M}{N(\alpha_2-\alpha_1)}\\ \nonumber & -\frac{\log_2\binom{N}{\alpha_1 N}+\log_2\binom{\beta N-1}{\alpha_1 N-1}+\alpha_1 N\log_2 M}{N(\alpha_2-\alpha_1)}.
\end{align}
Since $K \in \mathbb{Z}^{+}$, the minimum value of  $\Delta \alpha=\alpha_2-\alpha_2$ is $\Delta \alpha = \frac{1}{N}$.  By substituting $\alpha_2=\alpha_1+\frac{1}{N}$ into \eqref{eq:eq7}, the optimal value of $\alpha$ can be obtained when $\frac{\Delta \eta}{\Delta \alpha}=0$. By arranging this equation, one can obtain classical quadratic program and this classical program yields the solution in \eqref{eq:eq8}.\footnote{Note that the quadratic program produces two solutions. However, only the solution in \eqref{eq:eq8} satisfies the constraints in \eqref{eq:eqopt}.} 

\end{proof}

\begin{prop}
As $N \to \infty$, the asymptotic SE of OFDM-IM can be given as
\begin{align}\label{eq:eq12}
    \eta &\sim H(\alpha)+ \log_2 (\beta^{\beta}/\alpha^{\alpha}) -(\beta-\alpha)\log_2(\beta - \alpha)+\alpha \log_2 M
\end{align}
where $H(\alpha)$ is the entropy function of $\alpha$. 
\end{prop}

\begin{proof}
The proof is simply obtained by using the Stirling's formula\footnote{ $N! \sim \sqrt{2\pi N} (N/e)^N$ when $N \to \infty$.} in the SE and ignoring the expressions that approach zero as $N \to \infty$. 
\end{proof}

\section{Numerical Results}\label{sec:secV}
In this section, we provide numerical BER and SE results. In BER figures, ``OFDM-ICM $(N, K, I, M)$'' is the proposed OFDM-ICM scheme having $K$ activated subcarriers out of $N$ in each block, choosing the energies of the activated subcarriers according to the compositions of $I$ into $K$ parts and employing $M$-PSK modulation on the activated subcarriers, whereas ``OFDM-ICM $(N, K, I)$'' stands for the special case of OFDM-ICM that embeds information into only indices of subcarriers and compositions of an integer $I$ with $K$ parts.  ``OFDM-IM $(N, K, M)$'' signifies the conventional OFDM-IM scheme in which $K$ out of $N$ subcarriers are activated to send $M$-PSK modulated symbols in each block. Finally, ``OFDM-WCM $(N, I, \lambda)$, Alg. 1'' stands for the OFDM-WCM scheme that employs Algorithm 1 in \cite{Yarkin2020comp} and uses $\lambda$ to alter the SE as well as integer compositions of $I$ with $N$ parts to decide the energies of $N$ subcarriers, whereas ``OFDM-CM $(N, I, M)$'' is the OFDM-CM scheme that determines the energies of $N$ subcarriers according to the composition of an integer $I$ with $N$ parts and carries $M$-ary PSK symbols on each subcarrier. For the simulated schemes in this section, we pick $E_T=N$, thus the average energy per subcarrier is assumed to be one.

\begin{figure}[t!]
		\centering
		\includegraphics[width=7.5cm,height=5.5cm]{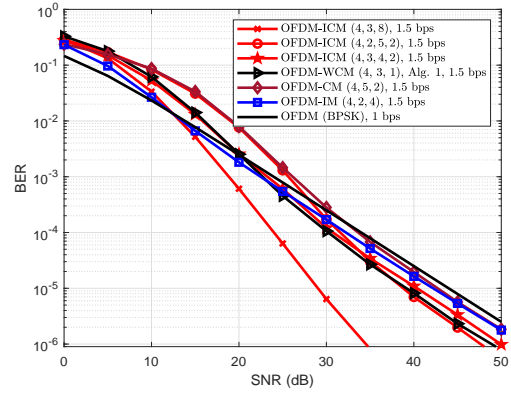}
		\caption{BER comparison of OFDM-ICM with OFDM-WCM, OFDM-CM, OFDM-IM, and OFDM when $N=4$, $K\in \big\{2, 3\big\}$, $I \in \big\{4, 5, 8\big\}$, $M\in \big\{2, 4\big\}$ and $\lambda=1$.}
		\label{fig:fig1}
\end{figure}

In Fig. \ref{fig:fig1}, we compare the BER performance of OFDM-ICM with that of OFDM-WCM, OFDM-CM, OFDM-IM, and OFDM. Here, except for OFDM (BPSK), all schemes achieve the same SE, which is 1.5 bits per subcarrier (bps). As seen from the figure, OFDM-ICM $(4, 3, 8)$ considerably outperforms all other schemes by introducing an additional diversity gain. Moreover, the BER performance of OFDM-ICM $(4, 3, 4, 2)$ is very close to that of OFDM-WCM $(4, 3, 1)$, Alg. 1 and these schemes outperform the OFDM-CM, OFDM-IM,  and OFDM schemes, especially at high signal-to-noise ratio (SNR), since they produce more symbols whose minimum Hamming distance is two. However, since the minimum Euclidean distance for the overall codebook of OFDM-ICM and OFDM-WCM is less than that of OFDM-IM, they are outperformed by OFDM-IM at low SNR. Due to the same reason, OFDM-ICM $(4, 2, 5, 2)$ is outperformed by all other schemes except for OFDM-CM $(4, 5, 2)$ at low SNR; however, it outperforms all other schemes except for OFDM-ICM $(4, 3, 8)$ at high SNR.    

\begin{figure}[t!]
		\centering
		\includegraphics[width=7.5cm,height=5.5cm]{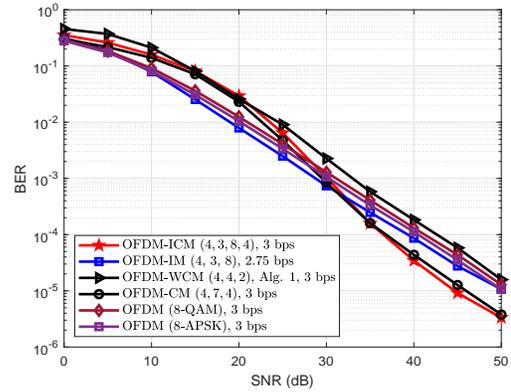}
		\caption{BER comparison of OFDM-ICM with OFDM-IM, OFDM-WCM, OFDM-CM, and OFDM when $N=4$, $K=3$, $I\in \big\{4, 7, 8\big\}$, $M \in \big\{4, 8\big\}$, and $\lambda=2$.}
		\label{fig:fig2}
\end{figure}

    Fig. \ref{fig:fig2} compares the BER performance of OFDM-ICM with that of OFDM-IM, OFDM-WCM, OFDM-CM, and OFDM.\footnote{For the curve related to OFDM (8-APSK), we use the optimum 8-ary amplitude and phase shift keying (APSK) design in \cite{Lucky62,Thomas74}. In this design, constellation consists of two circles (inner and outer) with four signals on each of them. The radii of the inner and outer circles are $r_1=1/\sqrt{2}\simeq 0.707$ and $r_2=1.366$, respectively.} As observed from the figure, OFDM-ICM $(4, 3, 8, 4)$ and OFDM-CM $(4, 7, 4)$ exhibit very close BER performance and they outperform all other schemes, especially at high SNR. Compared to Fig. \ref{fig:fig1}, the effectiveness of the OFDM-WCM scheme against the other schemes substantially degrades as it is required to employ higher order modulation in this scheme to achieve 3 bps. However, the OFDM-ICM scheme preserves its effectiveness in both data rates. It is also worth mentioning that OFDM-ICM $(4, 3, 8, 4)$ considerably outperforms OFDM (8-QAM) and OFDM (8-APSK) that use the energies of the constellation symbols for encoding like the proposed scheme.

\begin{figure}[t!]
		\centering
		\includegraphics[width=7.5cm,height=5.5cm]{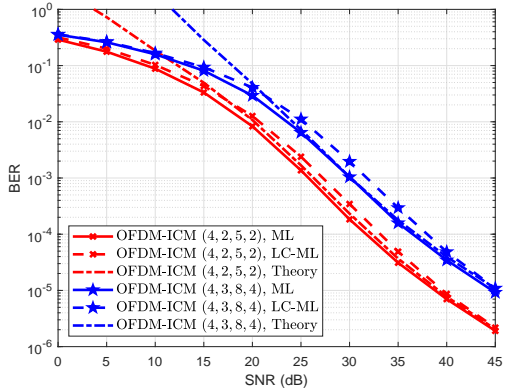}
		\caption{BER comparison of the proposed LC-ML detector with the optimum ML detector for  OFDM-ICM.}
		\label{fig:fig3}
\end{figure}

Fig. \ref{fig:fig3} demonstrates the BER performance of the proposed LC-ML detector of the OFDM-ICM scheme compared to the optimum ML detector when $N=4$, $K\in \big\{2, 3\big\}$, $I \in \big\{5, 8\big\}$, and $M \in \big\{ 2, 4\big\}$. As seen from the figure, the performance of
the proposed LC-ML detector is very close to that of the
optimum ML detector, especially at low and high SNR. On the other hand, curves with extensions ``Theory'' in the legend
relate to the theoretical upper-bound results for the proposed scheme. As observed from the figure, upper-bound curves are consistent with computer simulations, especially at high SNR.

\begin{figure}[t!]
		\centering
		\includegraphics[width=7.5cm,height=5.5cm]{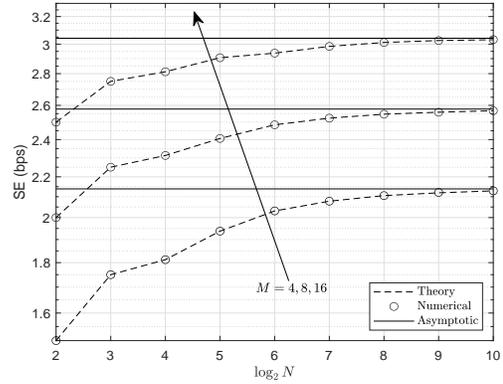}
		\caption{Comparison of theoretical, numerical and asymptotic SE results for OFDM-ICM when $\beta=0.5$ and $M\in \big\{4, 8, 16\big\}$.}
		\label{fig:figse}
\end{figure}

In Fig. \ref{fig:figse}, we compare theoretical, numerical and asymptotic SE results for OFDM-ICM when $\beta=0.5$ and $M\in \big\{4, 8, 16\big\}$.\footnote{Here, curves with extensions ``Theory'' in the legend are obtained by substituting \eqref{eq:eq9} into \eqref{eq:eq2last}, whereas those with with extensions ``Numerical'' in the legend are obtained by substituting all possible values of $K$ into \eqref{eq:eq2last} and choosing the maximum SE. On the other hand, curves with extensions ``Asymptotic'' in the legend are obtained by substituting \eqref{eq:eq8} into \eqref{eq:eq12}.} Fig. \ref{fig:figse} verifies the validity of the SE expressions in Section \ref{sec:secivb} since the theoretical and numerical SE results match perfectly and they approach asymptotic SE as $N$ increases. On the other hand, the SE is considerably improved when we increase $N$. When $N=4$, the OFDM-ICM schemes can achieve only $70.14 \%$, $77.60 \%$, and $82.20 \%$ of the asymptotic SE for $M=2, 4,$ and 8, respectively, whereas these percentages increase to $81.82 \%$,  $87.30 \%$, and $90.41 \%$ for $M=2, 4,$ and 8, respectively, when $N=8$.

\section{Conclusion}\label{sec:secVI}
In this paper, we proposed a novel modulation concept, which we call ICM. We depicted a practical OFDM-ICM scheme that is capable of improving the SE of OFDM-IM by embedding data into the compositions of an integer and using energy levels to identify them. We demonstrated through simulations and theoretical calculations that the proposed scheme can provide noteworthy improvements compared to OFDM, OFDM-IM,  OFDM-WCM, and OFDM-CM. 

As future work, the proposed scheme could be generalized
by utilizing in-phase and quadrature dimensions of the modulation symbols as in \cite{Fan2015} and space time coding with coordinate interleaving could be combined with the proposed scheme as in \cite{Basar2015} to obtain additional diversity gain.
\bibliographystyle{IEEEtran}
\bibliography{journal_paper}

\end{document}